\documentclass{article}
\usepackage{amsmath}
\usepackage{color}
\usepackage{graphicx,psfrag,epsf}
\usepackage{enumerate}
\usepackage{url} 
\usepackage{natbib}
\usepackage{caption}
\usepackage{subcaption}
\usepackage{float}
\usepackage{amssymb}  
\usepackage{amsthm}
\usepackage{mathrsfs} 
\usepackage{stmaryrd} 
\usepackage{txfonts} 
\newtheorem{theorem}{Theorem}

\newtheorem{proposition}[theorem]{Proposition}

\newcommand{\B}{{\mathbb{B}}}
\newcommand{\SSS}{{\mathbb{S}}}
\newcommand{\N}{{\mathbb{N}}}
\newcommand{\z}{{\overline{z}}}

\usepackage{soul}
\usepackage[section]{placeins}
\usepackage[normalem]{ulem}

\begin{document}




\title{A Likelihood Ratio Detector for Identifying Within-Perimeter Computer Network Attacks.}

\author{Justin Grana \and David Wolpert \and Joshua Neil \and Dongping Xie \and Tanmoy Bhattacharya \and Russel Bent}


\maketitle
\begin{abstract}
The rapid detection of attackers within firewalls of
enterprise computer networks is of paramount importance.  
Anomaly detectors address this problem by quantifying deviations from 
baseline statistical models of normal network behavior and signaling
an intrusion when the observed data deviates significantly from the
baseline model.  However, many anomaly detectors do not take
into account plausible attacker behavior. As a result, anomaly
detectors are prone to a large number of false positives due to
unusual but benign activity.  This paper first introduces a stochastic
model of attacker behavior which is motivated by real world attacker
traversal. 
Then, we  develop a likelihood \emph{ratio} detector that
compares the probability of observed network behavior under normal
conditions against the case when an attacker has possibly compromised
a subset of hosts within the network.  Since the likelihood ratio
detector requires  
integrating over the time  each host becomes compromised, we illustrate how to
use Monte Carlo methods to compute the requisite integral.  We then
present Receiver Operating Characteristic (ROC) curves for various
network parameterizations that show for any rate of true positives,
the rate of false positives for the likelihood ratio detector is no
higher than  
that of a simple anomaly detector and is often lower.  We conclude by
demonstrating the superiority of the proposed likelihood ratio
detector when the  network topologies and parameterizations are
extracted from real-world networks.

\end{abstract}


\section{Introduction}
Many existing systems designed to detect intrusions into computer
networks monitor data streams only at the perimeter of the network. In
addition, many network intrusion detection systems, such as snort 
~\citep{roesch1999snort}, are signature based, meaning that every
communication entering or 
leaving the network is examined for matches to a database of signatures, or
indicators of compromise. At this point, the long list of
breaches to corporate networks \citep{targetbreach, anthembreach}
speaks loudly to the insufficiency of these methods.  Attackers are
able to innovate rapidly in order to avoid signature schemes, and
penetrate these perimeter systems seemingly at will.  Therefore, there
is a pressing need to identify attackers \emph{within} network
perimeters, and to do so using behavioral methods rather than
signatures.  


Anomaly detectors --- a model-based approach--- show promise in
detecting within-perimeter attacks.  In general, anomaly detectors
quantify  ``normal'' network behavior, and when observed behavior
significantly deviates from the baseline model, an intrusion is
signaled.  As a simple
  example, consider an anomaly detector that models a computer network  as 
a directed graph where nodes are users within a network and edges
represent a communication channel between users. The detector is then
calibrated such that it specifies the average rate of packet
transfer along each edge.  When the observed rate of packet transfers
is sufficiently different from the calibrated rate of packet transfers,
the detector signals an intrusion.

In practice, many reported anomalies end up being false,
reflecting behavior that is unusual but benign.   This is due in part
to an incomplete specification of normal network behavior in the null
hypothesis as well as the  difficulty in modeling and predicting the
behavior of humans that interact over the network.   There
are at least two approaches in addressing this issue. The first is to
improve the specification of the 
network under normal conditions.\footnote{
  This work is similar to reducing prediction error of network
  traffic.  See \citet{netpred}, \citet{endtoend} and
  \citet{timefreq}, for work that
  focuses on improving the modeling and prediction of normal network
  behavior.}  %
The second is to develop a model of
attacker behavior and compare the probability of the observed behavior
under the hypothesis that the network has been compromised against the
hypothesis that the network is functioning under normal conditions.
With an accurately specified attacker model, such an approach would
rule out benign but unusual activity as being malicious since it
is not consistent with attacker behavior.  Our work in this paper
takes the second approach.  More explicitly, to our knowledge this
paper is the   first to incorporate an exact parametric specification
of attacker   behavior into a    likelihood ratio detector for
identifying malicious traversal   activity within a network perimeter.


The challenge is in how to model the behavior of a network that has been
penetrated without pre-supposing  attacker methods, since these
methods evolve  rapidly. To see how this might be done, consider a
common attack conducted on an enterprise network.  First, a Phishing
email or set of emails, containing either a malicious attachment or a
link to an internet host serving malware, is sent to the target 
network. Click rates on Phishing emails, even after enterprise training is
conducted, can be as high as 50\% \citep{kumaraguru2009school}, providing a
high-confidence intrusion vector.

At this point, the firewall is penetrated and the attacker has control of an
initial host in the target network.  The attack is far from complete
since the initially compromised host is not the primary target of the
attacker.  Instead, the attacker seeks to penetrate the network and
access key servers. However, since credentials are typically required
to access these servers,  the attacker undergoes a process known as
{\em lateral movement} to move among hosts collecting these
credentials \citep{kent2013differentiating}. 
This means that there is a definite sequence in the movement of the attacker
across the network, from computers with low value (for any of the
goals of inserting malware, extracting data, or stealing credentials) to 
computers with higher value, such as data servers and active
directories. This will be true \emph{no matter what precise methods 
  the attacker uses} to achieve that movement.
As a result, the attacker's traversal  will leave a trace of
increasing network traffic going from low value computers to
progressively higher value ones. 
Therefore, an increase in network traffic along paths from low value
to high value nodes in a network can be used as the basis of a model
of network behavior once it has been penetrated.

The approach in this paper is parallel to that of \cite{strans} in
that we first propose a model of attacker behavior and a novel
detection criteria based on a likelihood ratio.  For various network
parameterizations, we simulate network activity in both the normal and compromised state.
We then employ receiver operating characteristic (ROC) curves to show
that the proposed likelihood ratio detector
outperforms a simple anomaly detector that does not exploit
information regarding the traversing nature of an
attack. In addition  we develop the Monte Carlo techniques  used to approximate the
relevant integrals in computing our proposed likelihood ratio.    Finally, we 
extract topologies and parameter data from real-world networks and
then simulate attacker behavior. The results  show that in real-world
networks, our proposed likelihood ratio detector is superior to the simple
anomaly detector.
%
%



\section{Background}
\label{background}

Model-based anomaly detection proceeds by modeling and estimating the
parameters, $\hat \theta$,  of a computer network  under normal conditions.
Next, given a dataset $\mathbf{D}$ under question, the likelihood of the
parameters given the data can be evaluated: $\mathcal{L}(\mathbf{\hat
  \theta} \mid \mathbf{D})$.  A generalized likelihood ratio test (GLRT)
can then be used to infer whether a more likely alternative
parameterization is present given data $\mathbf{D}$
\begin{equation*}
\label{glrt}
GLRT = \frac{\mathcal{L}(\mathbf{\hat \theta} \mid \mathbf{D})}
{\sup_{\mathbf{\theta}\in \mathbf{\Theta}}\mathcal{L}(\mathbf{
\theta} \mid \mathbf{D})}
\end{equation*}
where $\mathbf \Theta$ is an alternative parameter space. 
Typically, we choose what data $\mathbf{D}$ to collect in order to
facilitate statistical discovery of security breaches.  For example, the
network model under normal conditions might be a graph connecting computers
(nodes or hosts) with edges representing parameterized time-series of
traffic.  The data collected would then be communications between
nodes.  When the observed communication pattern is different from the
parameterized time-series, the anomaly detector would sound an alarm.
Additionally, since attacks typically cover multiple nodes and edges, subgraphs can be
used to group data from multiple nodes and edges into $\mathbf{D}$ for
increased detection power. Such graph based methods include
\cite{borgwardt2006pattern,eberle2010insider,neil2013scan,staniford1996grids,
  djidjev2011graph}.  


If we know that the attacker behaves according to a
specified alternative parameter vector, say $\bf \theta_A$, then the
uniformly most powerful test for rejecting the null hypothesis that no
attack is present is a likelihood ratio test where
$\bf{\theta}_A$ is used in the denominator.  That is, if we know the
attacker is behaving according to $\bf{\theta}_A$, the power of the
test is maximized when  using the test statistic 
\begin{equation}
\label{altglrt}
\overline{GLRT} = \frac{\mathcal{L}(\mathbf{\hat \theta} \mid \mathbf{D})}
{\mathcal{L}(\mathbf{ \theta_A} \mid \mathbf{D})}.
\end{equation}

However, the set of alternatives $\mathbf{\Theta}$ is typically
under-specified.  In other words, anomaly detectors do not specify
exact attacker behavior but simply restrict the parameter space of
alternatives.  A representative example of such
a detector is  the Modeled Attack Detector (MAD) given in
\cite{agganom1}. In their 
work, the authors consider the rate of  incoming traffic in order
to detect a Distributed Denial of Service (DDoS)  attack.  They assume that
that under normal conditions, the number of incoming connections can
be modeled by a Poisson 
distribution with average rate of messages per unit time of
$\lambda_B$.  The authors treat $\lambda_B$  as a known and
calibrated parameter.   
Therefore, given a sequence of incoming connections (i.e. one unit
of network traffic)  $\mathbf{D}=\{d_1, d_2... d_N\}$ 
per unit time interval,  the probability of observing
$\mathbf{D}$ under the hypothesis that $H_0$ = \textit{no attack is
  taking place}  is given by
\begin{align}
P(\mathbf{D}| H_0) =
\prod_{i=1}^N\frac{e^{-\lambda_B}\lambda_B^{d_i}}{d_i!}.
\label{poissonlikelihood}
\end{align}
The authors  assume that under a DDoS attack the network
receives additional malicious connections at  fixed, deterministic
time intervals but at an unknown rate.  If the 
rate was known, the probability of an observed sequence under the
hypothesis that $H_1$ =  
\textit{DDoS attack is occurring} is given by
\begin{align}
P(\mathbf{D}| H_1)
=\prod_{i=1}^N\frac{\lambda_B^{d_i-\lambda_m}e^{-\lambda_B}}{(d_i-\lambda_m)!}.
\label{poissonattackerlikelihood}
\end{align}
where $\lambda_m$ is the rate at which the network receives
malicious connections. 
In reality, $\lambda_m$---the rate under the alternative
hypothesis---is unknown so a simple
likelihood ratio test is unavailable.  Instead, the authors employ a 
GLRT which  is given by%
\begin{align}
LR = \frac{P(\mathbf{D} |
  H_0)}{\displaystyle\max_{\lambda_m}P(\mathbf{D} | H_1)}
\label{agglhood}
\end{align}
where the denominator is maximized when $\lambda_m = -\lambda_B 
+\frac{1}{N}\sum_{i=1}^Nd_i$.  
When $LR$ is less than a predefined threshold, the MAD
indicates an attack.


Although the model of \citet{agganom1} is concerned with DDoS attacks and
not within-perimeter attacks, it succinctly showcases the key elements
of anomaly detection that this paper intends to improve upon.  Like the
MAD, the work presented in this paper will use a likelihood ratio as
the detection criteria. 
However, unlike the MAD, this paper will focus on incorporating
\emph{exact} specification of attacker behavior into a likelihood
ratio detector for \emph{within-perimeter} detection to improve
detection accuracy.

The likelihood ratio detector we introduce in this paper is most similar to the
anomaly detector presented in \citet{neil2013scan}, known as PathScan.
In that work, 
communication channels between hosts are in either an active or
inactive state.  In the active state, it is assumed that
communications take place stochastically at a known rate.  In the
inactive state, no communication occurs.  
The anomaly detector does not observe the
state of the edge but instead only knows the probability of an edge
transitioning (transition parameter) between the active and inactive
states and observes communications between 
hosts during a moving time window. ``Attacker behavior'' is modeled as
an increase in the probability of an edge transitioning from an
inactive to an active state.  The authors then compute the
probability of an observed dataset under the hypothesis that the
transition parameter is equal to the calibrated transition parameter
and compare that likelihood to the likelihood under the maximum
likelihood estimate of the transition 
parameter.  When the value of the likelihood at the maximum likelihood
estimate is
sufficiently different from the value of the likelihood under the
calibrated transition parameter, PathScan indicates an intrusion.  Our
novel detector presented below is similar to PathScan in that it uses a
likelihood ratio to detect \emph{within perimeter} anomalies.
However, our detector  is different from PathScan in that it
explicitly captures the the fact that the attacker must traverse the
network when attempting to access key information located on various
nodes in the network.  Furthermore, our work shows how to incorporate
exact information of an attacker's strategy into the likelihood ratio.

It is important to note that likelihood and likelihood ratio
based approaches are not the only characterizations of model-based
anomaly detectors.  For example, \citet{lee2001information} focuses
on information theoretic measures for anomaly detection such as
entropy, conditional entropy, relative 
conditional entropy, information gain, and information cost.
From a learning-based approach, \citet{ryan1998intrusion}
uses artificial neural networks with supervised learning, based on the
back-propagation neural network, called Neural 
Network Intrusion Detector. However, the limitations to supervised
learning as applied to intrusion detection are well documented and are
discussed in \citet{sommer2010outside}.  

A sizable portion of the literature approaches anomaly detection from
the frequency domain.  The work of \citet{wavelet} uses principal
component analysis on wavelet transforms of network traffic to detect
anomalies in backbone networks.  Similarly, the work of \citet{strans}
uses the $S$-transform to convert network traffic data into the
frequency domain and then presents ROC curves that illustrate the
effectiveness of their approach.  

Another approach to anomaly detection is through the use of sequential
hypothesis tests. In a sequential hypothesis test, the test is applied
multiple times as the data are generated.  Data that are included in the
first hypothesis test can be included in hypothesis tests later in
time.  Two examples of sequential based 
hypothesis tests are \cite{sequentialpost} and \cite{cusum2}, who
employ  CUSUM charts to determine when the parameters that govern
network behavior (such as packet transfer rates between hosts) have
changed, which signals a potential intrusion. 


To pinpoint our work in the current literature, we note that our
proposed likelihood ratio detector can be classified as a
within-perimeter (like PathScan), likelihood ratio-based (like
\cite{agganom1}) network intrusion 
detector.  The analysis is performed in the time domain (like
\citet{agganom1} and \citet{neil2013scan}). It does \emph{not} use
sequential hypothesis tests (as in 
\citet{cusum2}), supervised learning (as in \citet{ryan1998intrusion}) 
or information theoretic quantities (as in
\citet{lee2001information}).  Additionally, our work enhances the
current state of the art likelihood ratio intrusion detection systems
in that it shows how to incorporate explicit traversal behavior into
the denominator of the likelihood ratio.  For a more detailed survey of
general anomaly detection, see the surveys of
\citet{chandola2009anomaly}, \citet{lee2000data},
\citet{garcia2009anomaly} and \citet{Kantas2009MC}.

\section{Model}
\label{model}
We model a computer network
as a directed graph, potentially with cycles, where each node (also
referred to as ``host'') represents a computer
or a human  inside the firewall. Each
node has an associated state. Examples of human nodes are users, system administrators, and hackers,
whose states can represent their knowledge, their strategies, etc.
Each directed  edge represents a potential communication directly connecting one node (human or computer)
to another node (human or computer). These edges have associated states, which
represent communication messages. So the computer network evolves according
to a Markov process across all possible joint states of every node and every edge. 


In this initial project, we only consider computer nodes, treating the human using a particular
computer as part of that computer. We also only consider those computers that are inside the firewall.
Each node can be in one of two states, ``normal'' or
``compromised''.  Similarly, each edge can be in one of two states, ``no message",
or ``message in transit". When a node is in a normal state,  it sends benign
messages%
\footnote{A message is also commonly called a ''connection.'' To
  maintain generality, the term message is used to avoid confusion
  with the concept of ``establishing a connection''.
}%
 along any of its directed edges according to an underlying Poisson
process  with a pre-specified rate.  To maintain generality, we
  do not define a message explicitly but only suggest that a message
  can be, among many other alternatives, a remote desktop protocol
  connection or file transfer protocol connection. When a node is
  compromised, it still sends benign 
messages at the same rate as when it is not compromised, but
now it superimposes malicious messages. These are generated according to
another Poisson process, with a much lower rate, thus effectively
increasing the Poisson rate for message emissions out of a compromised
node by a small amount. We use a stochastic model of attacker behavior
because in reality, it is impossible to perfectly predict each action of the
attacker.  In other words, we are not claiming that attackers
are strategically randomizing their actions.  Instead, all we are
claiming is that the only information afforded to the detectors is
that the attacker's behavior ---strategic or otherwise--- can be
described by a Poisson process.  
The task of network defense is to detect the small increases in the
message transmission rates  and decide whether they fall into a
pattern indicative of attack. 

For simplicity we assume
that if an edge from a compromised node to a non-compromised node  gains a new malicious message at time $t$,
then with probability $1.0$ the second node becomes compromised and
the new malicious message disappears immediately, 
leaving a trace on our net-monitoring equipment that that message traveled down that edge at $t$.
No node can become compromised spontaneously, and no node can become uncompromised.

\subsection{Definitions}

Let $G=(V,E)$ be the directed graph of a computer network where
$V = \{v_1, v_2 ... v_N\}$ is the set of nodes  
and $E$ is a set of directed edges that represent communication channels
  between nodes.
Let $\sigma \in {\mathbb{B}}^N$ denote the state of
all nodes in the network and $\sigma_{v_i}$ denote the state of node
$v_i$: $1$ representing the uncompromised state  and
$0$ representing the compromised state. 
The Markov process governing the computer network is parameterized by
the set of Poisson rates $\lambda \equiv \{(\lambda_{v,v',\sigma_v}) :
v, v' \in V, v' \ne v, \sigma_v \in \B\}$ giving the total rate at
which $v$ sends messages to $v'$ when $v$ is in state $\sigma_v$. 
We write the rate parameter for just emission of malicious message
from $v$ to $v'$ as $\Delta_{v, v'} \equiv \lambda_{v, v', 0} -
\lambda_{v, v', 1}$. A directed edge from $v$ to $v'$ exists if and
only if $\lambda_{v, v', s}>0$ for some $s \in \sigma_v$.

Suppose we observe the traffic on a net for a time interval $[0, T]$,
and denote by $(\tau, v, v')$ an observation that a message was added
at time $\tau$ to the edge from  $v$ to $v'$. The resulting dataset
$D = \{(\tau_i, v_i, v'_i)\}$, has $\tau_i \in [0, T]$ and each $(v_i, v'_i) \in V^2$. 
We assume that the observation process is noise-free,
i.e., that all messages are recorded and no spurious messages are. 

In time continuous processes like this, the probability that two nodes
get compromised at exactly the same time is precisely zero; as a
result, we can assume a strict time ordering among the compromised
nodes. For all $1 \le k \le N$ indicating the possible number of node
compromises that occur in $[0, T]$ (though others might occur later),
define $\SSS^k$ as the set of vectors $s \in V^{k}$ such that each
element of $\SSS^k$ uniquely defines an ordering of $k$ nodes that can
become compromised when the network perimeter has been penetrate.
Define $\SSS = \cup_{k=1}^{N} \SSS^k$.
 Also define the space $Z \equiv [0, T] \cup \{*\}$, whose elements
 are either a time of compromise in the observation window $[0,T]$, or a
 $*$ to indicate no compromise occurs in that interval.  We will mostly
 consider vectors $z = (z_{v_1}, z_{v_2}, ... z_{v_m}) \in Z^m$
 specifying the times of compromise of various nodes, and index
 components of the vectors $z$ by the nodes compromised (or not for
 \(*\)) at those times.  
So $z_{s_i}$ is
the time that \(s_i\), the $i$'th node to get compromised, gets compromised, or is a \(*\).


For each pair $(v, v')$, it will be useful to define an associated function
$\kappa_{v,v'}(z_v, D)$ that equals the number of messages recorded in $D$ as
going from $v$ to $v'$ before $z_v$, where for $z_v = *$, this is
interpreted as the total number of such messages 
in the observation window. Similarly define
${\underline{\kappa}}_{v,v'}(z_v, D)$ as the number of messages after
$v$ gets compromised, or $0$ if it never gets compromised. 

Finally, for any $k \leq \N$,  $\tau > 0$, we denote by $\tau^k$ the
subset of vectors $[0, \tau)^k$ whose indices are non-decreasing, i.e.,
  $x \in \tau^{\kappa} \Rightarrow x_i \le x_j \; \forall i, j > i$.
  If exactly \(k\) nodes get compromised in our observation window,
  elements of $\tau^k$ exhaust the possible sequence of times at which
  the nodes are compromised. 
In the discussion below, we use  ``$P( \ldots)$" to refer to either
probabilities or probability densities, with 
the context making the meaning clear.

Our likelihood ratio detector is based on comparing the probability of
$D$ under two different Poisson processes: 
one where there is no attack and one in which there is an attacker
at node $v_1$ at time $0$. An anomaly detector only considers the
first of these probabilities. 
Whether or not there is an attack, the probability of our dataset
conditioned on $z$ can be calculated as follows.  For each pair of
nodes \((v,v')\), \(\kappa_{v,v'}(z,D)\) messages flow from \(v\) to
\(v'\) in the period \([0,z_v]\) when the source was uncompromised
(and hence emitting at a rate \(\lambda_{v,v',1}\)), and if the node
gets compromised at $z_v$, then \(\underline\kappa_{v,v'}(z,D)\) is
the number of messags from $v$ to $v'$ in the period \([z_v, T]\).  As
all the emissions are independent, the net conditional probability
density is given by the product of these factors:     
%
\begin{align}
\label{eqn:ProbData}
P(D  \mid  z) =&  \prod_{v \in V}\prod_{v'\in V, v' \neq v}\bigg[ (1 - \delta_{z_v, *})\frac{e^{-{z_{v}}\lambda_{v, v', 1}}(z_{v}\lambda_{v, v', 1})^{\kappa_{v ,v'}(z, D)}}{\kappa_{v, v'}(z, D)!} \times \nonumber \\
&\frac{e^{-{(T-z_{v})}\lambda_{v, v', 0}}((T-z_{v})\lambda_{v, v', 0})^{{\underline{\kappa}}_{v ,v'}(z, D)}}{{\underline{\kappa}}_{v, v'}(z, D)!}  + \nonumber \\
&(\delta_{z_v, *}) \frac{e^{-{T}\lambda_{v, v', 1}}(T\lambda_{v, v', 1})^{\kappa_{v ,v'}(z, D)}}{\kappa_{v, v'}(z, D)!}\bigg]
\end{align}
where $\delta_{a, b}$ indicates the Kronecker delta function, and in
particular, $\delta_{z_v, *}$ equals 1 
if node $v$ is not compromised in the window $[0, T]$ and $0$ otherwise.
In the case of no attack, it is only the second summand that survives in every term, giving
the probability of $D$ given that there is no attack is
\begin{eqnarray}
\label{eqn:baseline}
P(D \mid  z = \vec{*}) =  
\prod_{v \in V}\prod_{v'\in V, v' \neq v}\frac{e^{-T\lambda_{v, v', 1}}(T\lambda_{v, v', 1})^{\kappa_{v ,v'}(z, D)}}{\kappa_{v, v'}(z, D)!} 
\end{eqnarray}
where $\vec{*}$ is the vector of all $*$'s. This is the only
probability considered by an anomaly detector, and is the 
first of the two probabilities considered by our likelihood ratio
detector. 

In our initial project, we assume that if an attacker is present in the observation window, at time $0$
they have compromised a particular node $v_1$ and no other node (In a full analysis we would
average over such infection times and the nodes where they occur according to some prior
probability, but for simplicity we ignore this extra step in this paper.). Accordingly, 
$z_{v_i} > 0 \;  \forall i > 1$  (whether there is an attacker or
not), and the second of the two probabilities 
we wish to compare is $P(D \mid z_{v_1} = 0)$. 

Unfortunately, our stochastic process model gives us $P(D \mid z)$, where
\(z\) specifies the times of infection of all nodes compromised in the
observation window.  To obtain, $P(D \mid z_{v_1} = 0)$, which only
specifies the  time of infection of the first node to be compromised,
we need to integrate over the other infection times, this yields the integral:
%
\begin{eqnarray}
&P(D \mid z_{v_1} = 0) = \nonumber \\
&\sum_{s \in \SSS} \int_{T^{|s|}} d \z \; P(D \mid \z, s) P(\z, s \mid \z_{1} = 0, s_1 = v_1) 
\label{eqn:integral}
\end{eqnarray}
\noindent 
The first probability density in equation \ref{eqn:integral}, $P(D \mid \z, s)$, is 
given by writing $z_{s_i} = \z_i$ for all $i \le |s|$, all other $z_v = *$,
and plugging
into equation \ref{eqn:ProbData}. (N.b., $\z$ is indexed by integers,
and $z$ by nodes.) 
The second probability density is the dirac measure
\(\delta_{s_1v_1}\delta(\bar z_1)\) if $|s| = 1$. For other $s$'s 
we can evaluate by iterating the Gillespie
algorithm~\citep{gillespie1977exact}: 
\begin{proposition}
As shorthand write ``$v \not\in s$" to mean $\forall i \le |s|, s_i \ne v$. 
For any $s,\z \in T^{|s|}$ where $|s| > 1$,
\begin{eqnarray*}
&P(\z, s \mid \z_{1} = 0, s_1 = v_1) = \nonumber \\
&\prod_{v \not \in s} e^{-  (T - z_{|s|}) \sum_{i \le |s|} \Delta_{s_i, v}}
\prod_{j=1}^{|s| - 1} \Delta'_{s, j+1} e^{\lambda'_{s, j} (\z_{j+1} - \z_j)}
\end{eqnarray*}
where
$\lambda'_{s, k} \equiv \sum_{i=1}^k \sum_{v \not\in \cup_{j=1}^k \{s_j\} : (s_i, v) \in E} \Delta_{s_i, v}$
\newline and $\Delta'_{s, k} \equiv \sum_{i=1}^{k-1} \Delta_{s_i, s_k}$ .
\label{prop:second_term}
\end{proposition}

\begin{proof}
To begin, expand 
\begin{align}
&P(\z, s  \mid  \z_{1}  = 0, s_1 = v_1) = P(\z_{2}, s_2  \mid  \z_{1}=0, s_1 = v_1) \; \times \nonumber \\
& P(\z_{3}, s_3  \mid \z_{2}, s_2,  \z_{1} = 0, s_1 = v_1) \; \times \ldots 
\label{eqn:expansion}
\end{align}

To evaluate the first term on the right hand side, $P(\z_{2}, s_2  \mid  \z_{1}=0, s_1 = v_1)$, 
expand the aggregate
rate of a malicious message leaving node $s_1$ if that node is
compromised as $\lambda'_{s, 1}$. 
The probability that node $s_1$ sends a malicious message to $s_2$
before sending one to any other node is  
$\frac{\Delta'_{s, 2}}{\lambda'_{s, 1}}$.  Also, the probability that
$s_1$ sends its first malicious message at time $\z_{2}$ is  
$\lambda'_{s, 1}e^{-\lambda'_{s_1}(\z_{2} - \z_1)}$.  Furthermore, the
time homogeneity of Poisson processes imply that the time to the first  
malicious message and the node to which it is sent are statistically
independent.  Therefore 
\begin{align}
&P(\z_{2}, s_2  \mid  \z_{1}=0, s_1 = v_1) = P(s_{2} \mid \z_2,  \z_{1}=0, s_1 = v_1) \; \times \nonumber \\ 
& P(\z_2  \mid \z_{1}=0, s_1 = v_1) \nonumber \\
&= P(s_{2} \mid  s_1 = v_1) P(\z_2  \mid \z_{1}=0, s_1 = v_1) \nonumber \\									
&= \Delta'_{s, 2} e^{-  \lambda'_{s, 1} (\z_{2} - \z_1)}
\label{eqn:FirstTime}
\end{align}

Next we similarly expand 
$P(\z_{3}, s_3  \mid \z_{2}, s_2,  \z_{1} = 0, s_1 = v_1) = P(s_3  \mid s_2,  s_1) P(\z_3  \mid \z_{2}, s_2, s_1) $. 
The set of edges that lead from
either $s_1$ or $s_2$ to some currently uncompromised node is 
$\cup_{v \ne s_1, s_2 \; : \; (s_1, v) \in E \text{ or } (s_2, v) \in E}$.
The sum of the malicious message rates of those edges is  $\lambda'_{s, 2}$
Therefore 
we have $P(s_3  \mid s_2,  s_1) = \Delta'_{s, 3} / \lambda'_{s, 2}$ and
$P(\z_3  \mid \z_{2}, s_2, s_1) =  \lambda'_{s, 2} e^{ \lambda'_{s, 2} (\z_3 - \z_2)}$,  so that
\begin{eqnarray*}
P(\z_{3}, s_3  \mid \z_{2}, s_2,  \z_{1} = 0, s_1 = v_1) &=&\Delta'_{s, 3} e^{- \lambda'_{s, 2} (\z_3 - \z_2)}
\end{eqnarray*}
Iterating through the remaining components of $s$ gives the second
product term on the right hand side in the claimed result. 
The first product term then arises by considering the time interval between $\z_{|s|}$ and $T$, 
during which no nodes $v$ not listed in $s$ receive a malicious
message from any of the nodes that are listed in $s$.
\end{proof}

To evaluate our likelihood ratio attack detector we need to plug the
results of Proposition \ref{prop:second_term} and
equation \ref{eqn:ProbData} into \ref{eqn:integral}, evaluate that
integral, and then divide by the 
likelihood given in equation \ref{eqn:baseline}.  
%
%


We acknowledge that in this specification, we are sacrificing some
model accuracy for computational tractability.  For example, we assume
we know which node an attacker initially infects upon infection, which
in real network scenarios is a plausible but unlikely scenario.  It
is possible to compute the probability of an attack when the initially
infected node is unknown by averaging the likelihoods for all
possibly infected initial nodes over some prior probability of infection.
Furthermore, we assume that in any given realization, there is at most 
one attacker present when in reality there can be any number of
attackers present.  Assuming the attackers act independently,
computing such a likelihood given two attackers are present would
involve elaborating our  model so that each 
node can be in state ``normal'', ``infected by attacker 1'',
``infected by attacker 2'' or ``infected by both attackers.''  We
would then need to integrate over the times the nodes change states.
Although this is possible and the requisite mathematics would be very
similar to the model presented above,  the essence of our
contribution is best illustrated with the simplest attacker model.
Future work will focus on adding in the intricate details of
real-world computer networks.

\subsection{Computational approximations}
\label{computational}

To use our likelihood ratio attack detector, we need to evaluate
equation \ref{eqn:integral}.
To do this we express it as the expected value of $P(D \mid \z, s)$ over all $\z$ and $s$, 
evaluated under the multivariate distribution $P(\z, s \mid \z_1=0, s=v_1)$.  We then reformulate
that expectation value, in a way that allows us to approximate it via
simple sampling Monte Carlo~\citep{roca04}.

To begin, we consider a new network $(V, E')$ created from our original network $(V, E)$ by adding
enough new edges to those in $E$ so that $V$ contains a (directed)
path from $v_1$ to every node in $V$. We leave
rates of both benign and malicious edges on all of the old edges (i.e., on all $e \in E \subseteq E'$) unchanged.
Define some strictly positive value  $\tilde{\lambda}$ so that both  $T \tilde{\lambda} N^2$ is
infinitesimal on the scale of 1 and so that $\tilde{\lambda}$ is infinitesimal on the scale of the smallest
rate in the original network. This ensures that the probability that any non-empty data set $D'$ generated 
with our new net has a message traverse one of the new edges before time $T$ is infinitesimal. 
This in turn means that the likelihood of any non-empty $D$ generated
with the new net is the same as its likelihood with the original net,  
whether we condition on there being an attacker or on there not being one. 

Recall that in the original network, there exists an edge between
nodes if and only if the rate of message transmission along the directed edge
is positive.  Furthermore, in the new net used to compute the value of the
likelihood, the only new edges are from $v_1$ to all nodes that did not
receive communication from $v_1$ in the original net.  Furthermore,
the rate of communication along these new edges is modeled as a
Poisson random variable with a positive but infinitesimally small rate
parameter.  Therefore,  we are still considering Poisson processes
with the new net, and all Poisson rates are greater than zero 
on all edges in the new net.
Combining this with
the fact that there is a path in $E'$ from $v_1$ to every node $v \in V$, we see that if $v_1$ is compromised
in the new net, then every node  in the new net gets compromised at
some finite time, with probability 1. This allows us to re-express
equation \ref{eqn:integral} as
\begin{eqnarray}
\int_{\infty^N} d\z \sum_{s \in \SSS} \delta_{|s|, R(\z)} P(D \mid \z, s) P(\z, s \mid \z_{1} = 0, s_1 = v_1)
\label{eqn:new_coordinates}
\end{eqnarray}
where $R(\z)$ is the number of components of $\z$ that are less than or equal $T$.
It is this expectation value that we approximate with simple sampling. 

Since it is the product  $\delta_{|s|, R(\z)} P(\z, s \mid \z_{1} = 0, s_1 = v_1)$ that is a normalized
distribution for this new integral's  regions of integration, we must sample from that.  To do this, we
iterate the expansion of \(P(\z, s \mid
\z_{1} = 0, s_1 = v_1)\) in equation \ref{eqn:expansion}, multiplying by the Kronecker delta function at each step.
Note that due to that Kronecker delta function, whenever
we reach an iteration $i$ where the sample $\z_i$ we generate is greater than $T$, 
before evaluating $P(D \mid \z, s)$ we first pad all components of $\z$ at $i$ or later to be ``*", and
set $s$ to be the current list. After evaluating $P(D \mid \z, s)$
for that $\z$ and $s$, we break out, and form a new sample of $P(\z, s \mid \z_{1} = 0, s_1 = v_1)$.

As an illustration, to sample the term $\delta_{|s|, R(\z)} P(\z_2, s_2 \mid \z_1 = 0, s_1 = v_1)$, we
%
first set $s_1 = v_1, \z_1 = 0$, and then sample $\lambda'_{s, 1}e^{-\lambda'_{\hat{s}_1}(\z_{2} - \z_1)}$ to get a value of $\z_2$.
If that $\z_2>T$, then we break and start generating a new sample.
Otherwise we sample $s_2$ according to $\frac{\Delta'_{\hat{s}, 2}}{\lambda'_{\hat{s}, 1}}$, and then iterate
to generate a sample of $P(\z_3, s_3 \mid \z_2, s_2, \z_1 = 0, s_1 = v_1)$.

\section{Experimental results}
\label{results}
We now present receiver operating characteristics (ROC) curves for
various network topologies, message transmission rates and 
observation windows. The results cover a wide range of typical network
structures and attacker behavior. 
These experimental results provide strong evidence that the likelihood
ratio detector
significantly outperforms state-of-the-art techniques based on anomaly
detection, irrespective of network topology and attacker strategy. 

An ROC curve is a two-dimensional plot that compares the true positive and false positive rates of 
a binary classifier. For a given threshold, the true positive rate is calculated as
$\frac{\text{True positives}}{\text{Total positives}}$ and the false
positive rate as $\frac{\text{False Positives}}{\text{Total
    Negatives}}$. These values are plotted for different threshold
choices to create a curve. When comparing ROCs, the \textit{higher}
the curve, the better.  


\begin{figure}
\centering
\includegraphics[width=.3\textwidth]{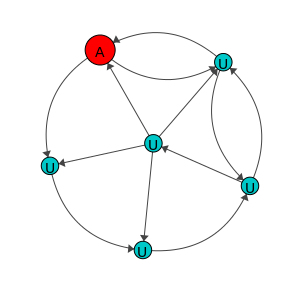}
\caption{A small, simple network.  The  node labeled ``A'' is the attacker.  The ``U'' nodes are the normal
users of the network.}
\label{fig:SmallStar}
\end{figure}

Our experiments are designed to represent stylized enterprise networks where the attacker has
penetrated the perimeter and can begin traversing the network. 
For example, the attacker may have compromised the computer of a credit card customer service
representative at a major bank.   However, he is not interested in the information available on that machine
but is interested in account information held on a central server.  As a result, he must emit messages from the
originally infected machine in order to gather credentials and elevate privileges until he has access to
valuable information.  This is however only one possible narrative and the experimental environment can
represent any number of enterprise attacks.  

In each of our experiments, we generated $400$ realizations of network activity over an observation window 
length $T$. There are $200$ cases with an attacker and $200$ cases without an attacker.
Recall that the dataset available to both of the detectors is a
  collection of triples of the form $\left(v_i, v'_i,\tau_i\right)$ indicating
  that a message traversed an edge from $v_i$ to $v'_i$ at time
  $\tau_i$.  From such a dataset it is possible to compute the
  likelihood ratio as well as the baseline likelihood detection
  criteria, which are both given in section \ref{model}.
For each realization, we compute the likelihood that the observed
message transmissions come from a system with no attacker 
and the likelihood that the observed message transmissions come from a
system with an attacker. The likelihood ratio classifier proposed in
the previous section is the ratio of the two
likelihoods. The ROC curve for anomaly detection uses only the likelihood
of no attacker as the classifier. 


In our first experiment, we analyze a small network as shown in
figure~\ref{fig:SmallStar} for  $T=1500$.  
The ``star'' formation  in figure \ref{fig:SmallStar}
  is a typical subgraph of many 
  computer networks.  More generally, the experiments in this section
  use topologies that are stylized representations of small sections
  of real-world subgraphs (such as a ``star'' or a ''caterpillar'').
  However, these topologies may not be realistic of all real-world
  networks.  To ensure that our results hold for more realistic,
  larger networks, the ``Real Data Experiments'' section extracts
  real-world network topologies  from the computer network at Los
  Alamos National Lab.

Normal message traffic
rates are set to $1$ and  
malicious message transmission was set to $3\%$ of the normal rate,
which models a relatively slow attack.  Under this scenario, the
attacker remains on the network for long periods of time and traverses
the network sporadically. \%
Note that in these initial experiments, we
  use homogeneous rates of benign message transmission.
  This allows us to investigate the effects on the ROC curves of
  changing other model parameters (such  as the observation window or
  rate of malicious  message  transmission) without having to consider
  how the effects depend on rate parameter heterogeneity.  
  Nevertheless, the ``Real Data Experiments''  section considers a high degree
  of rate heterogeneity and confirm  the superiority of the likelihood
  ratio detector.



\begin{figure}
\centering
\begin{minipage}[t]{0.45\textwidth}
\includegraphics[width=.9\textwidth]{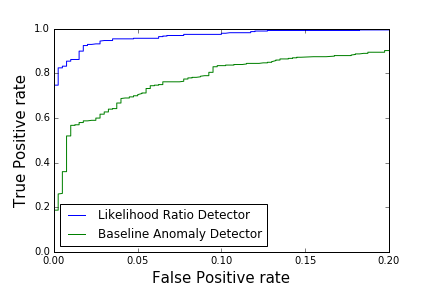}
\caption{ROC curves for the network topology shown in figure~\ref{fig:SmallStar}.}
\label{fig:roc1}
\end{minipage} \hspace{0.05\textwidth}
\begin{minipage}[t]{0.45\textwidth}
\includegraphics[width=.9\textwidth]{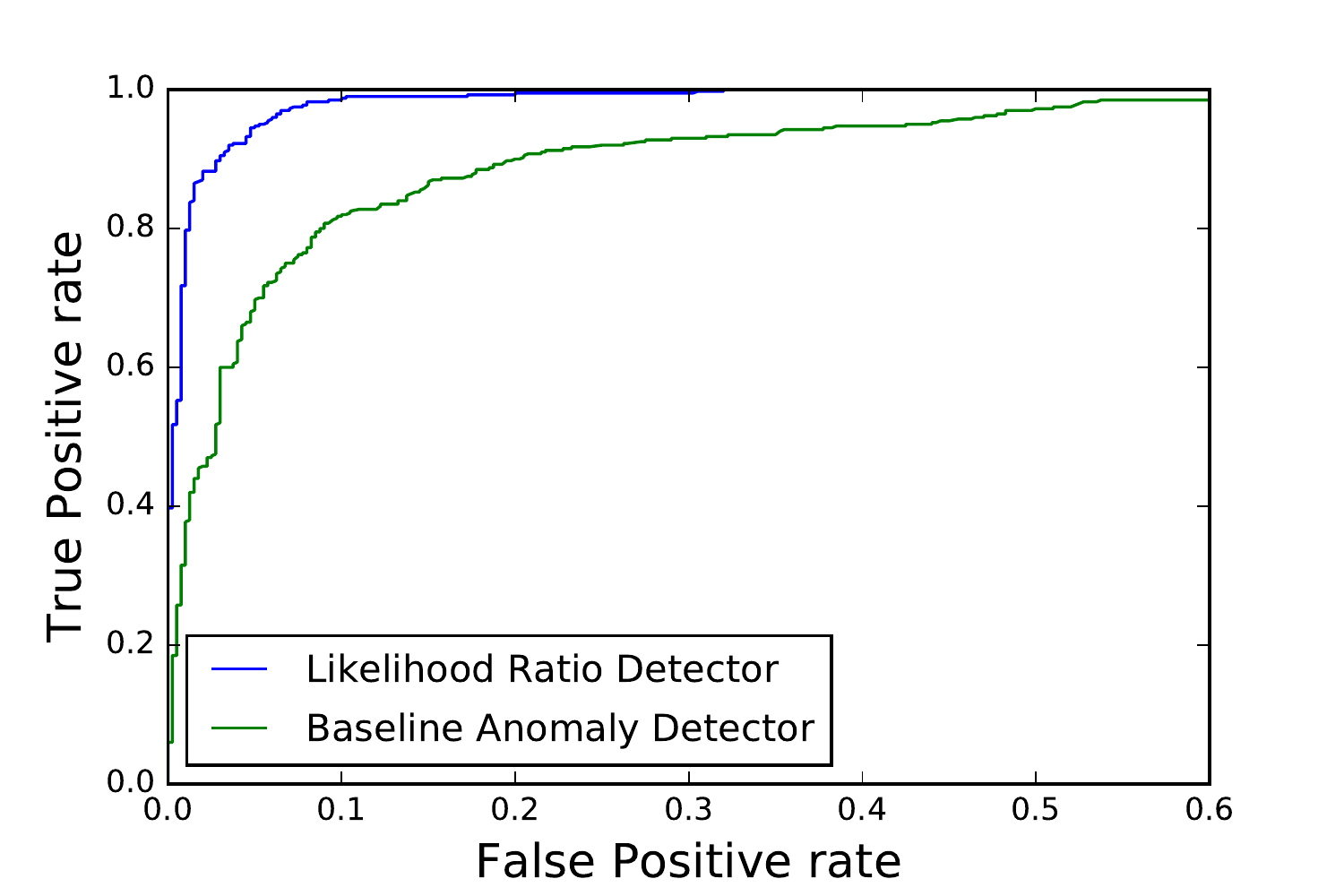}
\caption{ROC curve under the ``enter slowly, traverse quickly'' specification.}
\label{fig:roc2}
\end{minipage}
\end{figure}  
Figure~\ref{fig:roc1} shows that  the likelihood ratio detector outperforms 
the baseline anomaly detector.  That is, for any false
positive rate the likelihood ratio detector has a higher true positive
rate than the simple anomaly detector.  Furthermore, figure
\ref{fig:roc1} shows that the superiority of the likelihood ratio
detector is clear even for low values of the false positive rate.
This is important because in real world networks, limited resources
only allow security response teams to investigate a small number of
instances and thus must set the threshold to a relatively small rate
of false positives.  In other words, if the likelihood ratio
detector's superiority was only apparent for relatively high rates of
false positives, its benefits would not be realized in practice.
However, since the likelihood ratio detector's advantage is evident
for low false positive rates, the initial results suggest that it
would improve practical anomaly detection.  


\begin{figure}[htb]
\centering
\includegraphics[width=.45\textwidth]{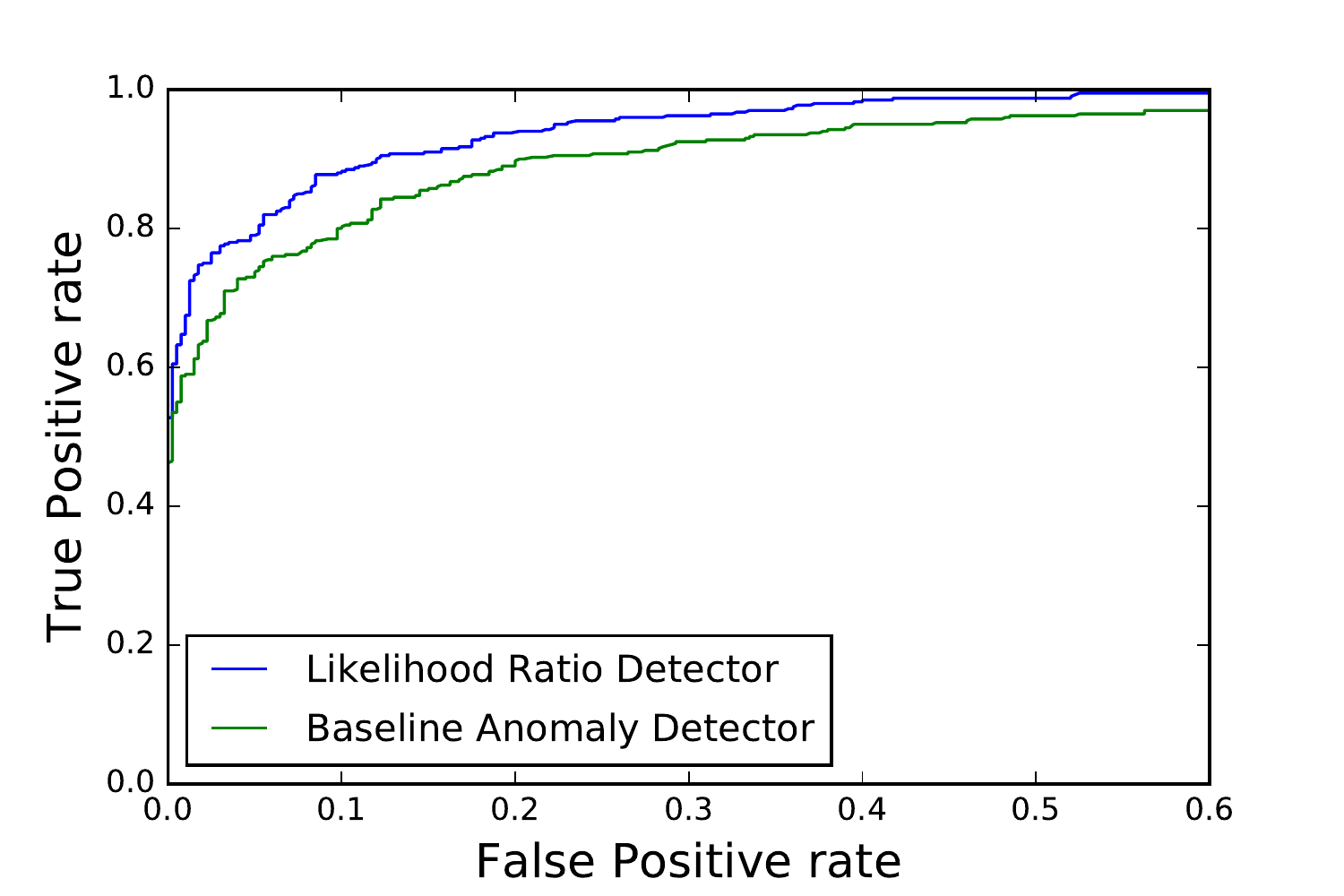}
\caption{ROC with $T=10$.}
\label{fig:roc3}
\end{figure}

In our second experiment, we consider the same network but adopt an
``enter slowly, traverse quickly'' strategy for the attacker and use
$T=400$. 
In this scenario, the attacker initially
sends messages at a rate of $3\%$ of the initial compromised node's normal transmission rate.  Once inside the
network, the attacker traverses the network rapidly by sending
messages at a rate of $6\%$ of the normal message rate.  
Figure~\ref{fig:roc2} once again shows that the likelihood 
ratio detector dominates the simple anomaly detector.  Like figure
\ref{fig:roc1}, the dominance of the likelihood ratio detector is
evident for low (0) levels of the false positive rate.

\begin{figure}
\centering
\includegraphics[width=.45\textwidth]{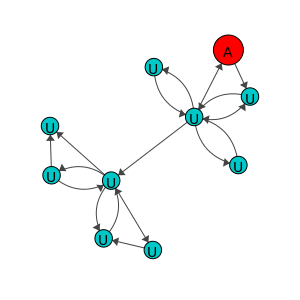}
\includegraphics[width=.45\textwidth]{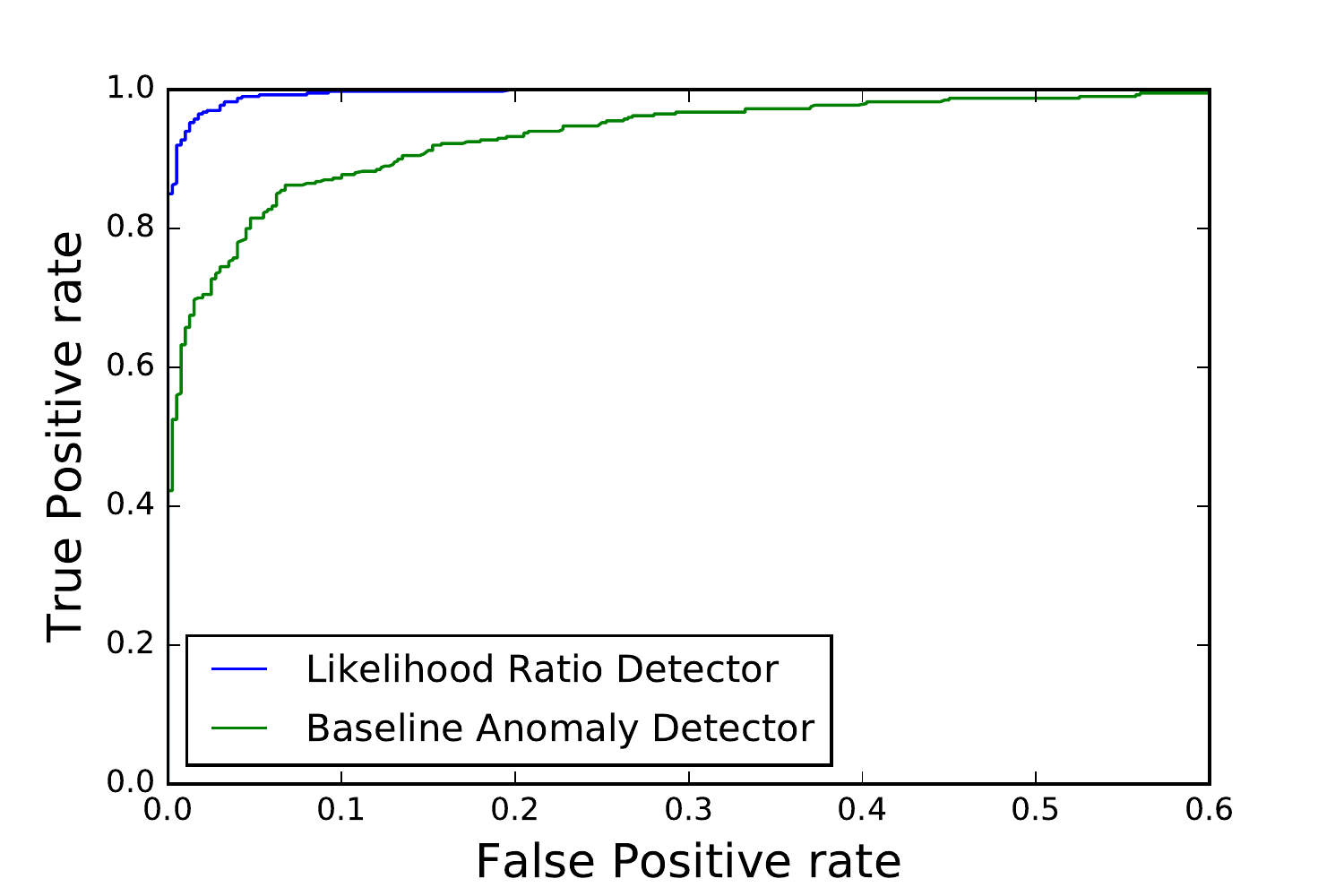}\\
\includegraphics[width=.45\textwidth]{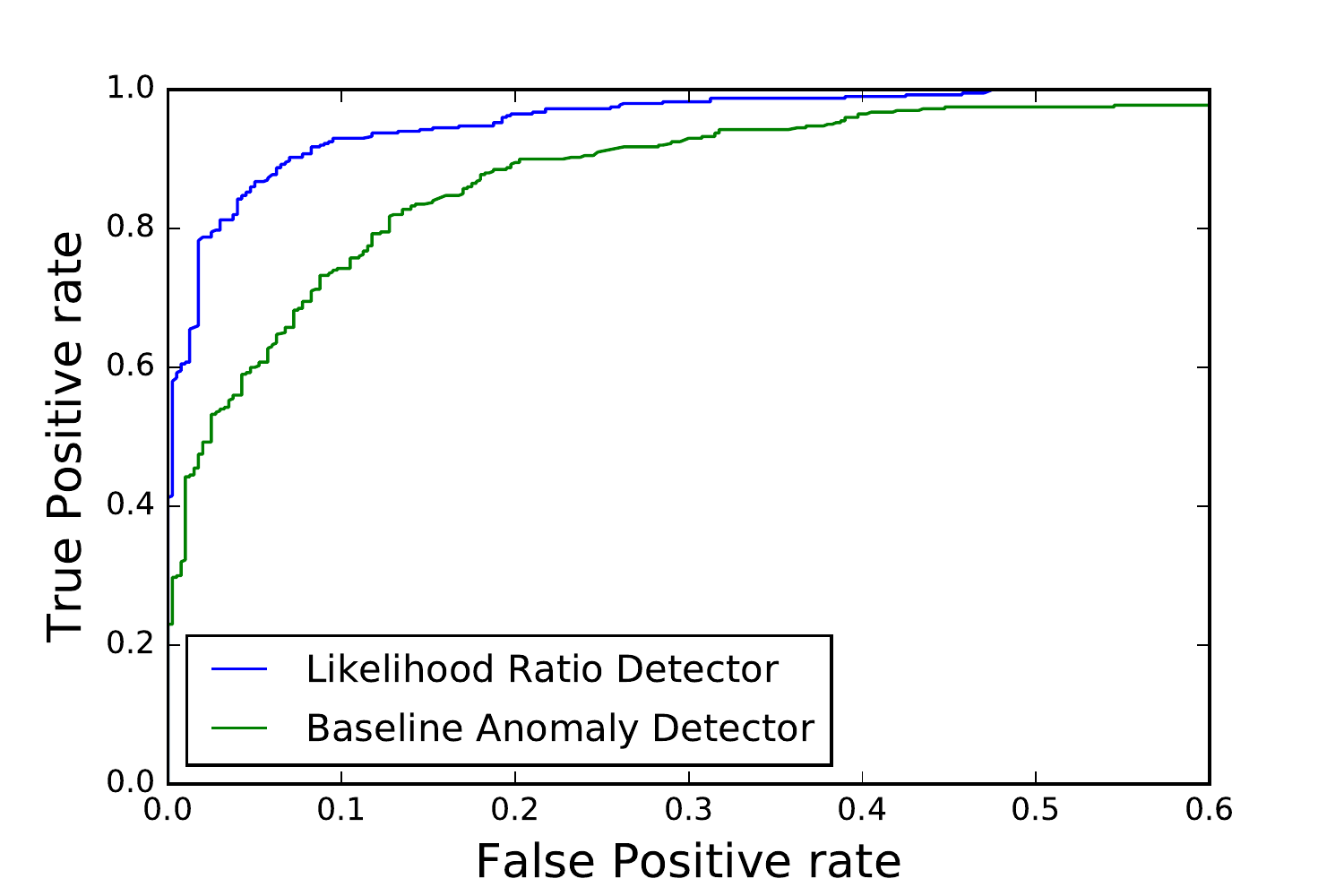}
\includegraphics[width=.45\textwidth]{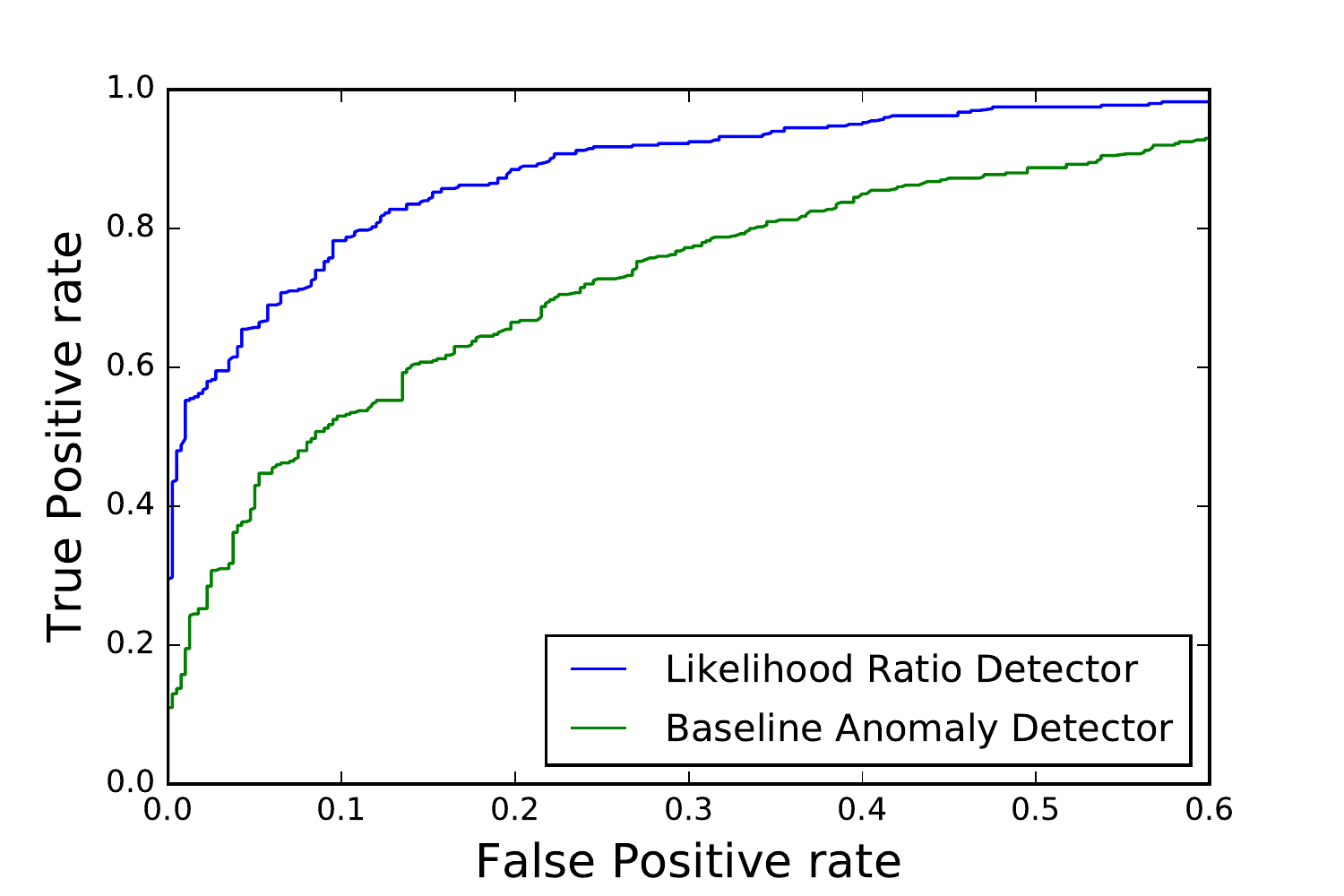}
\caption{Network topology (upper left) and ROC curves for various sets of parameters: 
 Upper right: $T=800$ and all rates 3\% normal rates. Lower left:
 $T=50$ and the rate of malicious messages out of the initially
 infected node is 3\% the normal message rate while malicious messages
 out of all other nodes is 6\% the benign rate.  
 Lower right: $T=10$ and malicious rates are 50\% of benign rates except for nodes
 with more than 2 outgoing edges.  For those nodes, the rate of
 malicious messages is 10\% of benign messages. }
\label{fig:SecondSet}
\end{figure}

In our third experiment, using the same network and $T=10$, we model
an attacker that makes no attempt to ``hide'' from the detectors, but
instead tries  
to traverse the network fast enough so that by the time an alarm
sounds, the network has already been compromised. Malicious message
rates are set to half the normal traffic rate. 
Figure~\ref{fig:roc3} provides the third validation
that the likelihood ratio detector performs better than the anomaly
detector and indicates that it is possible to detect the attacker
before he reaches his goal.

\begin{figure}[htp]
\centering
\includegraphics[width =.45\textwidth]{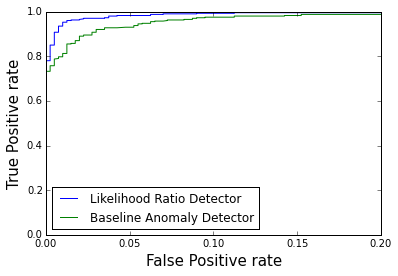}
\caption{ROC curve under the scenario where the attacker becomes more aggressive as he approaches the goal.}
\label{attackprogress}
\end{figure}
The fourth experiment employs the same three attacker heuristics but
this time  for the larger network described in figure
\ref{fig:SecondSet}.  All benign message transmission rates are once
again $1$ and the rate of malicious messages is given in the caption
of the figure.  The motivativation for this experiment  is
two-fold.  First, it is used to verify that the results of the
previous experiments were not specific to a certain network topology.
Secondly, the larger network contains more nodes and more edges and
thus requires more computational resources to compute the likelihood
ratio.  Therefore, the larger network serves as a test bed to see
if the likelihood ratio detector remains superior under a modest
scale increase.  
The ROC curves indicate that
the likelihood ratio detector outperforms the anomaly detector under
all three specifications in this case as well.   

A final test examines the performance of  the likelihood ratio
detector when the attacker's strategy is to increase  
in aggression as he approaches his end target.
Figure \ref{fig:misspecif}  shows a network where the  
attacker moves toward a goal (node 'G').  All nodes send normal
messages at rate $1$.  At node 'A',  
the attacker sends malicious messages at rate $.05$.   
For each subsequent node infection, the malicious message rate
increases by $.05$.   The observation window is set to $T=12$.  
Once again, figure \ref{attackprogress} shows that the likelihood ratio detector outperforms
the simple anomaly detector.  This is especially promising since the
advantage of the likelihood ratio detector is evident for very low
false alarm rates (note the $x$-axis ends at $.2$).  Additionally, it
shows that the likelihood ratio detector outperforms the simple
anomaly detector, even when the simple anomaly detector performs
relatively well.  (The anomaly detector detects 90\% of attacks with
a false positive rate of  about $.05$ in figure \ref{attackprogress}
compared to a false positive rate of $.5$ in the top right corner of
figure \ref{fig:roc2}). These results along with the previous
experiments suggest that the superiority of the likelihood ratio
detector exists regardless of the performance of the simple anomaly
detector.

\paragraph{Model misspecification}


The preceding results assumed that the rate of malicious message
transmission and attacker strategy is known. That is, the parameters
used to compute the likelihood ratio test statistic were the rates
that the attacker actually used.
In reality this is not usually the case.  Therefore, to test the performance of the likelihood ratio detector
in a more realistic scenario we assume that the attacker's strategy is
misspecified and test the performance of the likelihood ratio.
In this scenario, we use the network described in Figure~\ref{fig:misspecif} with $T=10$ and we set 
all normal message rates equal to $1$ and compute the likelihood ratio as if the rate of malicious 
messages is $.5$  on edges exiting any compromised nodes.  This model of the
attacker allows the detector to hedge for an attacker that could
choose any path from $A$ to $G$. 
In our experiments, the actual attacker only traverses the center path (using rate $.5$).
Figure~\ref{fig:ROCmis1} shows that even under misspecification,
the likelihood ratio detector is superior to the baseline anomaly
detector. 
In this specific example, the superiority is most pronounced
around a false positive rate of $.5$.  This is too high to be
practically significant but there are also modest performance
improvements for low false positive rates.  Furthermore, because the
likelihood ratio test uses incorrect parameters in the alternative
hypothesis, there is not theoretical guarantee that the likelihood
ratio outperforms the simple anomaly detector.  However, this example
shows that for any threshold, the likelihood ratio does at least as
well as the simple anomaly detector.

\begin{figure}
\centering
\begin{minipage}[b]{0.45\textwidth}
\centering
\includegraphics[width=.95\textwidth]{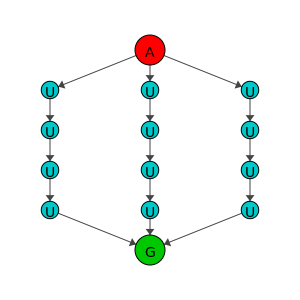}
\caption{Network topology with misspecified attacker strategy.}
\label{fig:misspecif}
\end{minipage}\hspace{0.05\textwidth}
\begin{minipage}[b]{0.45\textwidth}
\centering
\includegraphics[width=.95\textwidth]{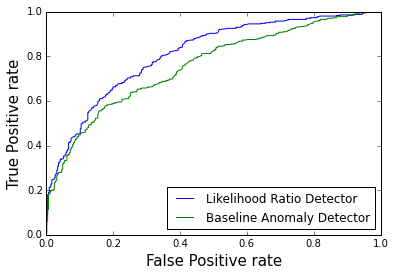}
\caption{ROC curve when the model allows the attacker to take all paths, but he only takes one.}
\label{fig:ROCmis1}
\end{minipage}
\end{figure}

To further test model misspecification, we consider the net in figure
\ref{fig:meanzeronoise}.  This topology is a stylized version of commonly
observed attacker behavior as noted in ~\cite{neil2013scan}  In short, an
attacker  sends messages to  all hosts connected to a host he has already 
compromised.  After exploring the hosts connected to the originally
compromised host, he will then occupy another host and then begin
exploring from that host.  This generates what is called a
``caterpillar'' pattern of malicious behavior.

\begin{figure}[ht]
\centering
\includegraphics[width=.45\textwidth]{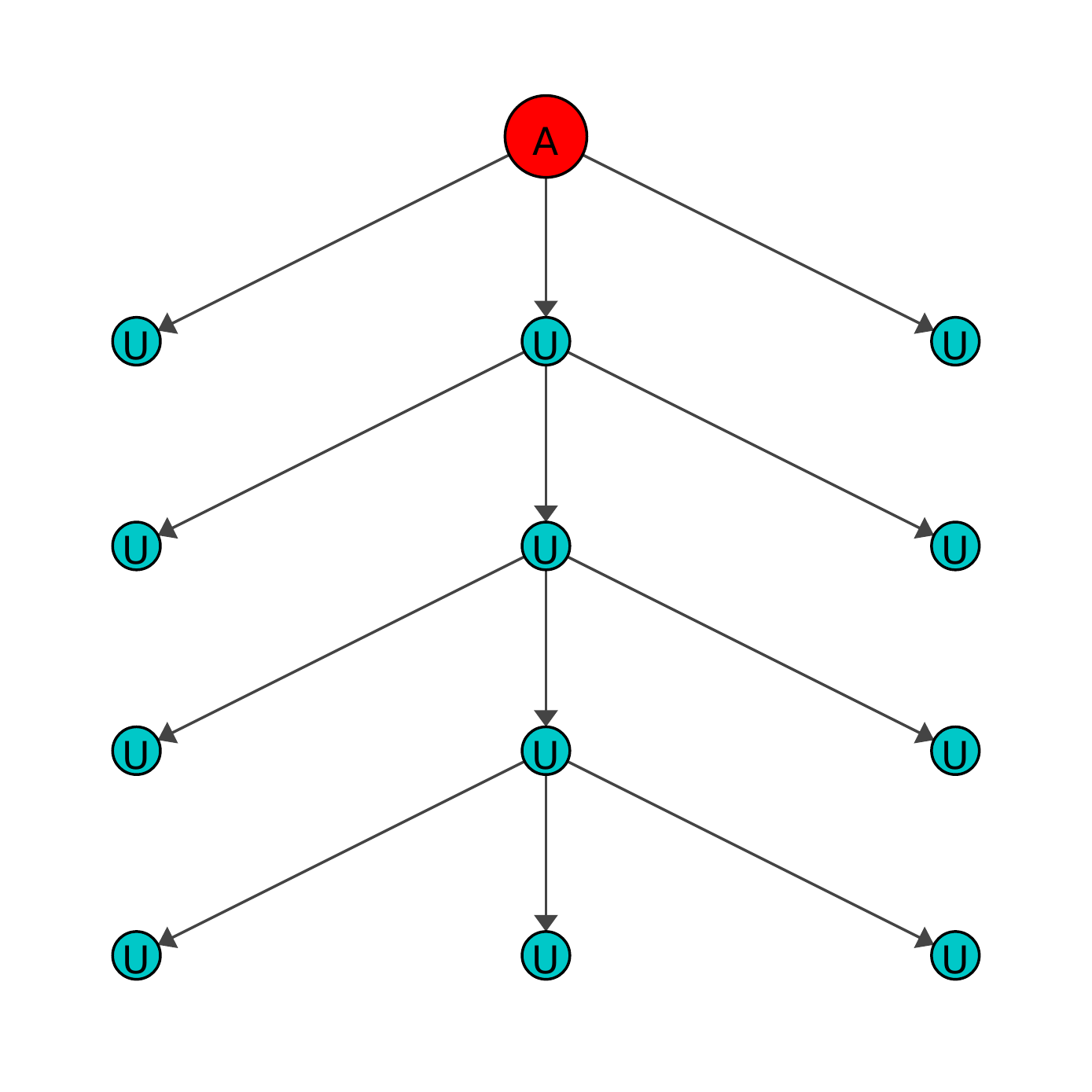}
\includegraphics[width=.45\textwidth]{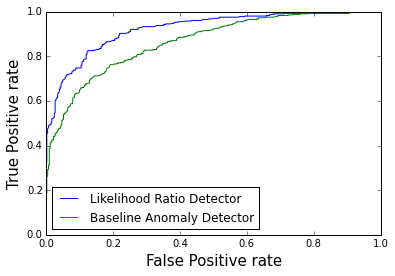}\\
\includegraphics[width=.45\textwidth]{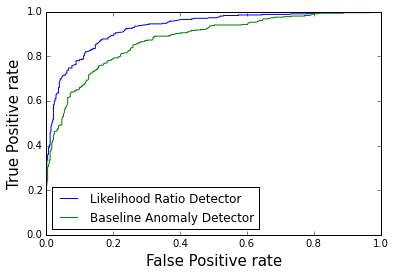}
\includegraphics[width=.45\textwidth]{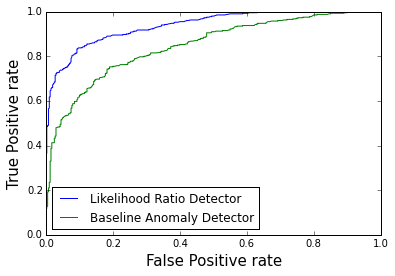}
\caption{Network topology and ROC curves for various levels of noise
  misspecification from upper right to lower right, the standard deviation of the
  noise is 10, 20 and 30 percent.}
\label{fig:meanzeronoise}
\end{figure}

The type of model misspecification analyzed here is one in which the specified
rate of attacker behavior is subject to non-systematic error.  Since
in most cases the attacker is a human with imperfect reasoning, it
would be unreasonable to expect to model the attacker's method
exactly.  Or on the other hand, it is not likely that an attacker
model will be able to perfectly specify what an attacker will do but
only be able to determine how attackers will behave \emph{on average}.
To analyze this type of model misspecification,  the attacker sends messages
along the center path and to the off center hosts in figure
\ref{fig:meanzeronoise}. (The rate along the 
center path is $.5$ and off center is $.25$.  The rate of background
messages is $1$ for all hosts).  However, the rate at which the
likelihood ratio models the attacker is subject to mean-zero, Gaussian
noise with standard deviations of $10$, $20$ and $30$ percent of the actual 
attacker rate. 

Figure \ref{fig:meanzeronoise} shows that even under non-systematic errors in
the attacker rates, the likelihood ratio detector still
performs better than the simple anomaly detector.  Interestingly,
increasing the standard deviation of the noise seems to have little to
no effect on the performance of the likelihood ratio detector.  The
slight difference in performance is likely a result of stochasticity
of the underlying data generating process.

\begin{figure}[ht]
\centering
\includegraphics[width =.45 \textwidth]{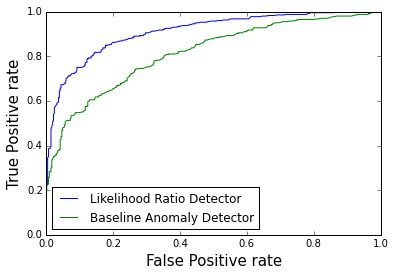}
\caption{ROC curve where the attacker is assumed to only take the
  center path.}
\label{centeronly}
\end{figure}

Another test of model misspecification in the caterpillar net
is the converse of the misspecification that generated figure
\ref{fig:ROCmis1}.  For that experiment, we modeled the attacker as
traversing any and all 
possible paths in the network where in reality he only traversed one
path.  However in this example, we calculate the likelihoods as if the
attacker only took the center path of the caterpillar net but in reality,
he was actually exploring the left and right hosts.  

This type of model misspecification has at least two real world
interpretations.  The first is that the model misspecifies the
attacker as being more intelligent than he actually is.  For example,
if the attacker's goal is to traverse the caterpillar to the last
center host, he would not gain anything from exploring the nodes off
of the center path.  However, if the attacker does not know the
topology of the network and the location of the valuable information,
he will explore the off-center nodes in the network.  Therefore, this
version of model misspecification models the attacker as more
intelligent and efficient than he actually behaves.  

Another interpretation is of practical concern.  The strength of a
likelihood ratio detector will --- to some degree --- depend on how accurate
the model of the attacker is.  However, the attacker is a goal driven
agent and therefore his behavior would be the result of an
optimization problem or even the solution to a fixed point problem.  
Such problems in complex 
environments such as a computer network  can become intractable with
only a small number of hosts (especially fixed point
solutions). As a result, it might only be feasible to consider attacker behavior
at hosts with the most sensitive information.  Therefore, it is
imperative for the likelihood ratio detector to be effective, even
when it only considers malicious activity along a subset of edges in
the network.  

Figure ~\ref{centeronly}  shows that the likelihood ratio detector
does outperform the simple anomaly detector even when the likelihood
ratio assumes the attacker only takes the center path.  It is
interesting to note that under this specification, the likelihood
ratio is being handicapped because the malicious activity off of the
center path enters multiplicatively and equally in the numerator and the
denominator of the likelihood ratio and therefore cancels.  However,
the malicious activity off of the center path does provide evidence of
an attack when using the simple anomaly detector.  In other words,
since the likelihood ratio assumes the rate of malicious message
transmission off of the center path is $0$, all messages off of the
center path are considered as normal background behavior, regardless
of how abnormal the messages seem.  On the other hand, the simple
anomaly detector considers traffic along all edges and therefore
abnormal activity off of the center path would push the simple
anomaly classification criteria toward the threshold.

\subsection{Real Network Experiments}
\label{sec:real}
In all of the previous experiments, the network topology and message
transmission rates were selected to be a stylized representation of
real world networks.  To ensure that the results hold for more
realistic networks, this section extracts network topologies
and estimates message rates from data collected from an active
computer network.  The data used are the Los
Alamos National Lab (LANL) ``User-Computer authentication associations in
time'' \citep{cyberdata}.  The data contain time-stamped user-computer
authentication logs.  For example, one data point is given by:
\begin{center}
\vspace{6pt}
\textit{U1, C5, 1}
\vspace{6pt}
\end{center}
\noindent which represents a user with an anonymous identification
number $U1$ logging into a host with identification number $C5$ at time $1$.
The data are taken over a period of nine months and the unit of time is
seconds.  The time stamps are given in whole numbers and therefore the
data are binned into seconds.  

Concurrent connection attempts by the same user constitute
communication between two hosts.%
\footnote{In fact, this is an assumption.  It might not always be true
  that concurrent connection attempts constitute communication between
  two hosts, consultation with the author of the dataset suggests that
  this is almost always the case.}%
Suppose, for example,  $U1$ logs
into $C1$ and then in 
the \emph{same second}  $U1$ logs into $C2$, then
a communication is inferred, which is then modeled as a directed  edge
from $C1$ and $C2$.  The network topology is 
constructed in this way because of the following logging artifact:
If  $U1$ is logged onto $C1$ and successfully logs onto $C2$,  the
authentication log will record $2$ entries.  The first will be an
authentication log at $C1$; user $1$ must be authenticated at $C1$ in
order to log onto other hosts.  The second will be a recording at $C2$
that notes a successful authentication into $C2$. In the language of
the preceding section, an authentication attempt is a
``message.'' 

Since the dataset covers nine months of authentication requests, which
includes $11,362$ users and $22,284$ computers, the computational approach
only allows for the analysis of  a subset of the network at a relatively
small slice in time.  To do this, only one hour intervals in the
first month of the dataset are considered.  However, due to computational limitations
only 
the elements of the network that see the most traffic are included.
This modeling choice is motivated by two generally observed patterns.  
First, the part of the network with the most traffic is
where it will be most difficult to detect an attacker.  In other
words, an attacker's rate is always relatively low.  Therefore, if the
background rate between hosts is relatively high, the attacker's signal to noise
ratio is low, thus rendering detection more difficult.
On the other hand, it
is easier for an attacker to go undetected when there are many users
logging into a specific machine.   The second motivation is that high
traffic areas are likely to be the location of valuable information
and hence be a likely location to deploy a novel attack detector.
The reason for this is that central servers typically contain the most
valuable information and therefore often receive the most traffic.

To extract a subgraph for analysis,  a random hour
in the first month of the data is sampled and then a full
graph is created based on the
method described above.  From the full graph,  any nodes (and
corresponding edges into and out of the node) that do not have more
than $15$ incoming messages and $15$ outgoing messages are removed.  This results
in several disconnected subgraphs.  The final graph  is
the largest connected component.  The estimate of the
per hour message transmission  rate is simply the one hour message
count.  That  rate is used to \emph{simulate} the background messages.  That is,
actual message times are not used as the observed data. Instead,  the
observed messages are used to estimate a Poisson rate to generate the times of
normal message transmission.  The distribution of per minute message
rates between nodes in the two extracted topologies is given in figure
\ref{realrate}.  The figure indicates that there is significant rate
heterogeneity with the highest rate being more than $10$ times that of
the lowest rate in each experiment.  
\begin{figure}
  \centering
  \includegraphics[width=.9\textwidth]{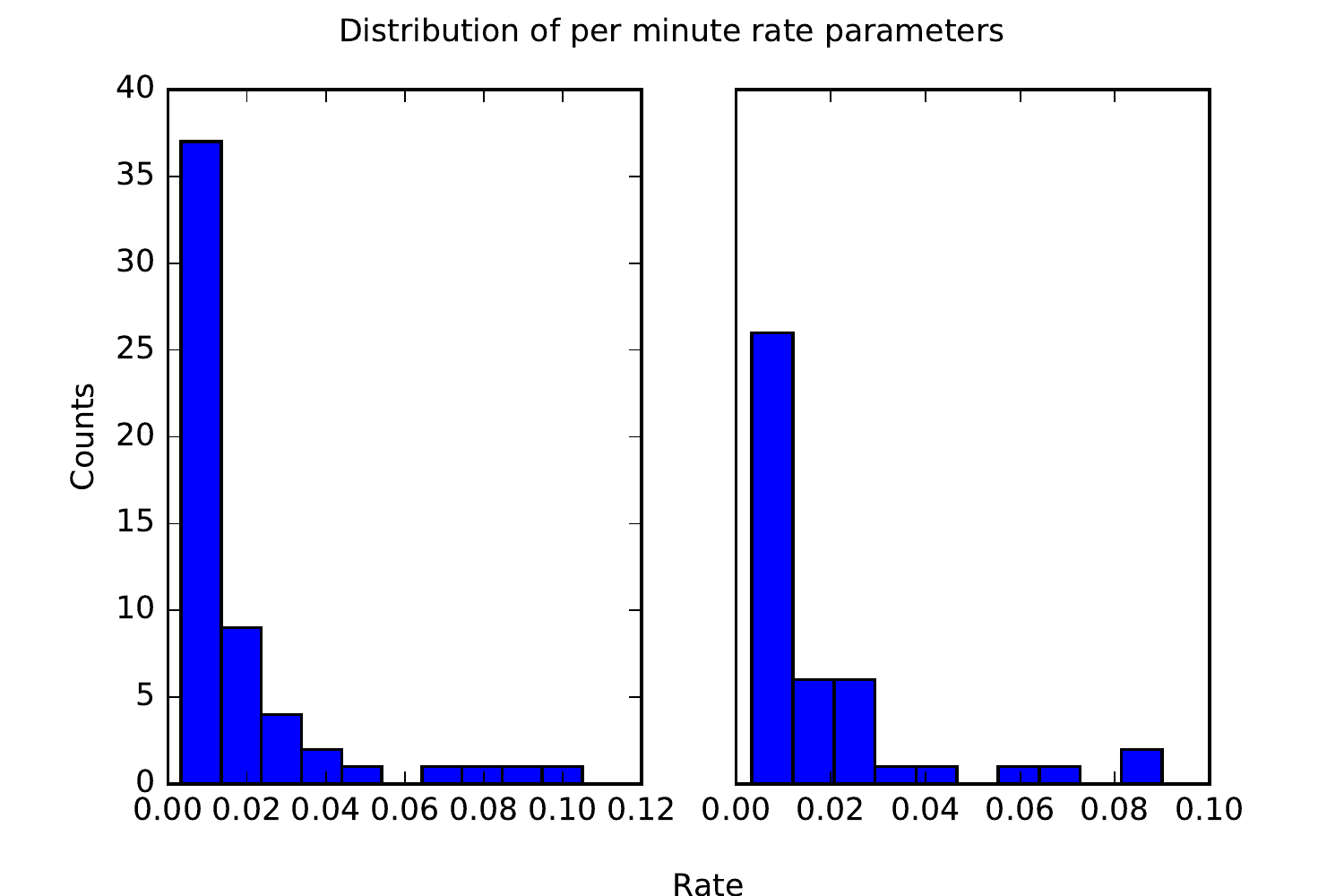}
  \caption{ The distribution of Poisson rates extracted from the LANL
    network.}
  \label{realrate}
  \end{figure}

Although using the actual message times
would be a more accurate representation of the real data,  time 
heterogeneity of the rate constants  would significantly increase both the model
sophistication as well as the computational requirement
to carry out the likelihood calculation in this initial framework.  In
other words,
although in a small time window, message transmission counts can be
described by a Poisson or negative binomial distribution, this rate
changes throughout the working day (and the weekend).  Therefore, the
likelihood ratio in section  \ref{model} would have to incorporate a
stochastic process that defines the behavior of the background
message transmission rates, thus adding another dimension to the
integral calculation in the denominator of the likelihood.  Such
modeling remains a topic of future work.  Furthermore, since the
ROC analysis of the real data requires $800$ simulations of the network,
there are not enough data points to carry out an accurate ROC
analysis. In other words, even if  it was known that on Tuesdays from 12PM to
1PM, the rate of message transmission was always the same, the data
only contains 36 such occasions which is not sufficient for ROC
analysis.  

Since the dataset does not contain attack data,  the rate of
malicious messages transmission is set to 10\% of the background rate.  Since
it is also not known which node might be initially compromised,
multiple experiments are performed in which the initially compromised
node is uniformly sampled.  

\paragraph{Real Data Results and Discussion}
The first subgraph analyzed is depicted in figure \ref{fig:real1}.
The numbered nodes  correspond to the attacker's
starting point for each of the experiments.  The ROC curves that
compare the likelihood ratio detector to the simple anomaly 
detector are given in figure \ref{fig:rocreal1}.  The second subgraph
analyzed is depicted in figure ~\ref{fig:real2}  and the associated
ROC curves are in figure \ref{fig:rocreal1}.  In all of the results,
the network is simulated for $T=60$ minutes.  

\begin{figure}[!htb]
\centering
\includegraphics[width =.8 \textwidth]{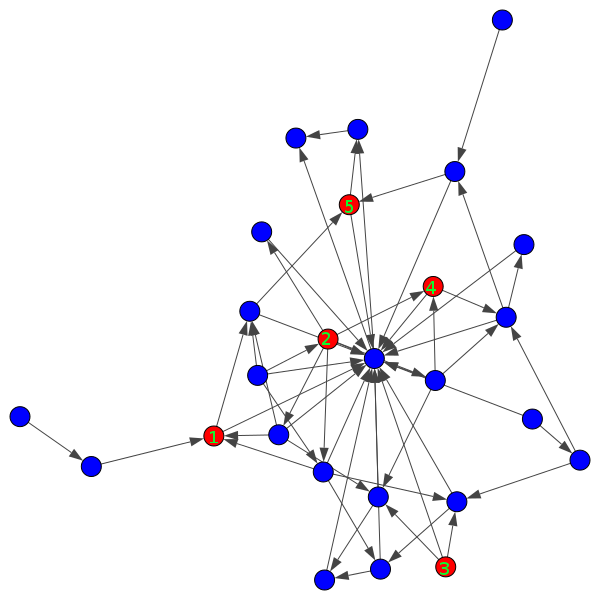}
\caption{First network topology extracted from LANL network.}
\label{fig:real1}
\end{figure}
\begin{figure}[!htb]
\centering
\includegraphics[width=.4\textwidth]{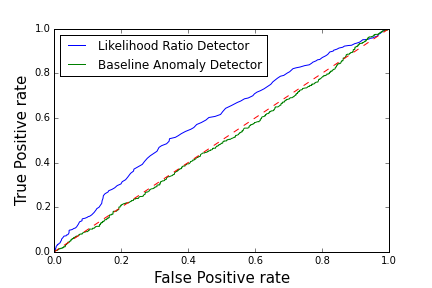}
\includegraphics[width=.4\textwidth]{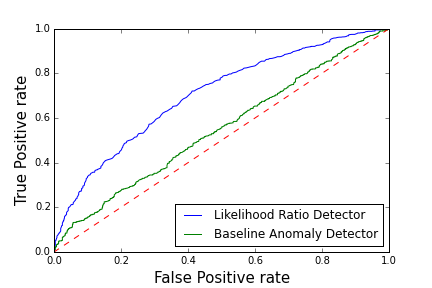}\\
\includegraphics[width=.4\textwidth]{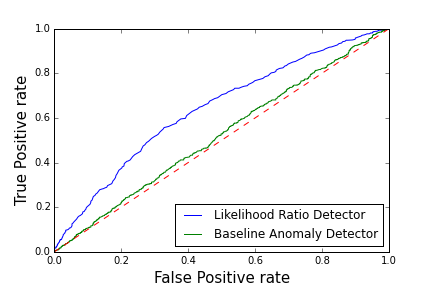}
\includegraphics[width=.4\textwidth]{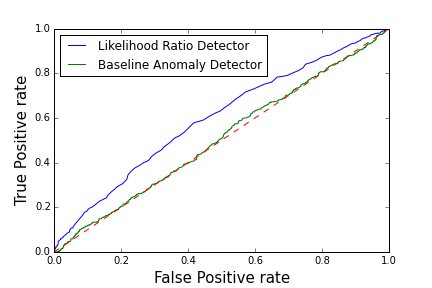}\\
\includegraphics[width=.4\textwidth]{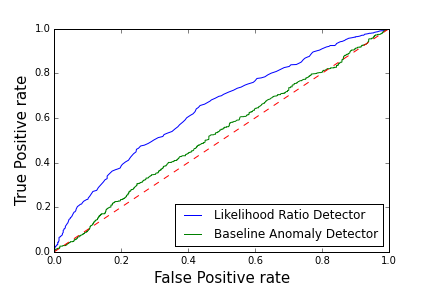} 
\caption{ROC results for real network topologies.  The plots from left
to right represent attacker starting points 1-5 as numbered in figure ~\ref{fig:real1}. }
\label{fig:rocreal1}
\end{figure}
\begin{figure}[!htb]
\centering
\includegraphics[width =.8 \textwidth]{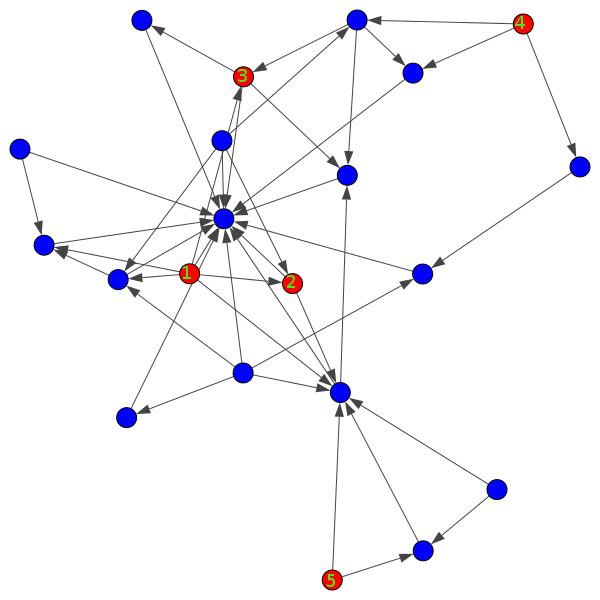}
\caption{Second network topology extracted from LANL network.}
\label{fig:real2}
\end{figure}
\begin{figure}[!htb]
\centering
\includegraphics[width=.4\textwidth]{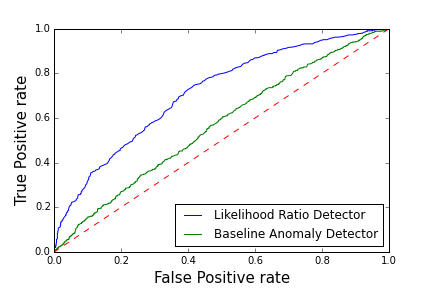}
\includegraphics[width=.4\textwidth]{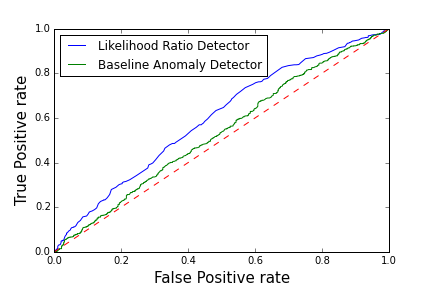} \\
\includegraphics[width=.4\textwidth]{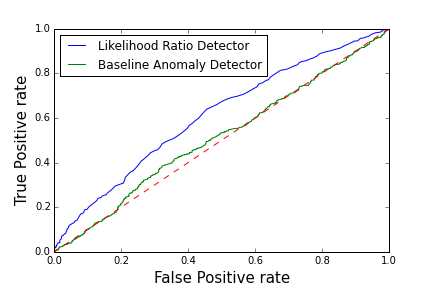}
\includegraphics[width=.4\textwidth]{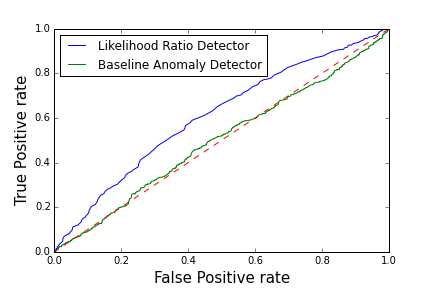}\\
\includegraphics[width=.4\textwidth]{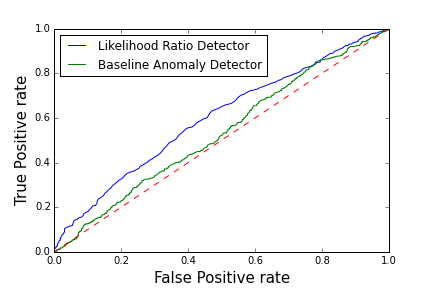} 
\caption{ROC results for real network topologies.  The plots from left
to right represent attacker starting points 1-5 as numbered in figure
~\ref{fig:real2}.} 
\label{fig:rocreal2}
\end{figure}

From the figures, it is clear that the likelihood ratio detector is
superior at all threshold levels to the simple anomaly detector.
Interestingly,  the overall performance of
both detectors seem to perform worse in the real network than in the
simulated network.  The reason for this is that in the real network,
only a subset of the nodes can become compromised.  For example, in
figure ~\ref{fig:real2} where the attacker initially has compromised node
numbered 2, only three of the 21 nodes can possibly become compromised.
Therefore, the signal of the attacker behavior is low compared to the
background noise generated by the entire network.  In the previous
experiments, it was possible for a greater proportion of nodes to
become compromised.  Nevertheless, the likelihood ratio detector still
outperforms the simple anomaly detector in all scenarios.  

It is also interesting to note that even when the simple anomaly
detector performs no better than random guessing (in the figures, this
is the case when the line representing the simple anomaly detector is
almost the same as the dashed 
line which represents the performance of uniform random guessing), the likelihood ratio
detector offers a slight advantage over random guessing.  The reason
for this is due to the definition of a likelihood.  More specifically,
the likelihood  can be relatively low due to the presence
of an attacker.  However, the likelihood can also be low due to an
atypical lack of activity.  In this case, the low signal to noise
ratio coupled with the simple anomaly detector sometimes classifies
low traffic as anomalous causes the detector to perform no better than
random.  Note that  anomalous low
traffic --- a phenomenon almost certainly not indicative of an attacker
--- does not affect the likelihood ratio's classification performance
as this type of behavior will appear in both the numerator and
denominator and cancel.  This can be solved in the anomaly detector by
restricting the parameter space to test for rate increases, as is done
in \cite{neil2013scan}.  However, the attacker model
proposed incorporates more information than just rate increase
testing, and therefore performs better.
\begin{table}[!htb]
\centering
\resizebox{\columnwidth}{!}{%
\begin{tabular}{|l|l|l|l|l|}
\hline
\multicolumn{5}{|c|}{Summary of Results}                                                                                                                                                                                                                                                                                                                                                                                                                                \\ \hline
 \begin{tabular}[c]{@{}l@{}}Graph\\Topology\end{tabular}                                                                      & ROC Curve                                                                                       & \begin{tabular}[c]{@{}l@{}}Misspecified \\ Rates?(Y/N)\end{tabular} & \begin{tabular}[c]{@{}l@{}}Clear Dominance \\ at  False  Positive\\ Rate $<.05$?\\ (Y/N)\end{tabular} & \begin{tabular}[c]{@{}l@{}}Likelihood Ratio\\ No worse  at all\\  False Positive\\ Rates?\end{tabular} \\ \hline
 \ref{fig:SmallStar}                                                            & \begin{tabular}[c]{@{}l@{}}Fig. \ref{fig:roc1},  \ref{fig:roc2} \ref{fig:roc3}\end{tabular} & N                                                                   & Y                                                                                                   & Y                                                                                                   \\ \hline
\begin{tabular}[c]{@{}l@{}}Fig. \ref{fig:SecondSet} \\ (Top Left)\end{tabular}      & \begin{tabular}[c]{@{}l@{}}Fig. \ref{fig:SecondSet} (Top Right, \\ Bottom Right)\end{tabular}   & N                                                                   & Y                                                                                                   & Y                                                                                                   \\ \hline
\begin{tabular}[c]{@{}l@{}}Fig. \ref{fig:SecondSet}\\  (Top Left)\end{tabular}      & \begin{tabular}[c]{@{}l@{}}Fig. \ref{fig:SecondSet} \\
(Bottom Left)\end{tabular}                                                          & N                                                                   & N                                                                                                   & Y                                                                                                   \\ \hline
Fig. \ref{fig:misspecif}                                                            & Fig \ref{attackprogress}                                                                        & N                                                                   & Y                                                                                                   & Y                                                                                                   \\ \hline
Fig. \ref{fig:misspecif}                                                            & Fig. \ref{fig:ROCmis1}                                                                              & Y                                                                   & N                                                                                                   & Y                                                                                                   \\ \hline
\begin{tabular}[c]{@{}l@{}}Fig. \ref{fig:meanzeronoise} \\ (Top Right)\end{tabular} & Fig. \ref{fig:meanzeronoise} (All)                                                              & Y                                                                   & Y                                                                                                   & Y                                                                                                   \\ \hline
\begin{tabular}[c]{@{}l@{}} Fig. \ref{fig:meanzeronoise}\\ (Top Right)\end{tabular}                                            & Fig. \ref{centeronly}                                                                           & Y                                                                   & Y                                                                                                   & Y                                                                                                   \\ \hline
Fig. \ref{fig:real1}                                                                & Fig. \ref{fig:rocreal1} (All)                                                                   & N                                                                   & Y                                                                                                   & Y                                                                                                   \\ \hline
Fig. \ref{fig:real2}                                                                & Fig. \ref{fig:rocreal2} (All)                                                                   & N                                                                   & Y                                                                                                   & Y                                                                                                   \\ \hline
\end{tabular}
}
\caption{Results of ROC experiments: The first column gives
  the figure representing the graph topology.  The second column gives
the ROC curve for the experiment.  The third column notes if the
likelihood ratio detector is clearly superior at a false positive rate
of .05 and less.  The final column notes whether the likelihood ratio detector
performs at least as good as the simple anomaly detector for all false
positive rates and better at one or more false positive rates.}
\label{ressum}
\end{table}
Table \ref{ressum} provides a brief summary of all experiments.
\section{Future work}
\label{conclusion}

There are many ways that our  model of the behavior of
a network can be extended. Most straightforwardly, by incorporating a
loss function for incorrect alerts and specifying a 
prior probability of there being an attacker, we should be
able to construct a Bayesian decision-theoretic extension of our anomaly detector.

Other future work involves applying our modeling approach for more realistic
networks.
This will likely require us to consider other approaches to evaluating our likelihood, e.g.,
importance sampling or MCMC, rather than simple sampling. Indeed, it may
even be possible to do closed form evaluation of our integrals, using the
Laplace convolution theorem~\citep{wolpert2013estimating}.

It should be possible to use our model to make predictions for \emph{any} 
scenario in which humans interact with technical systems
in continuous time, and our observational data is limited. In other
future work we will apply our models to make predictions in such
scenarios. (See also~\citet{wobo14b}.) 
This should allow us to address 
any statistical question concerning such scenarios, not just for anomaly detection.

Finally, we have started to extend our approach
to model not just a single human interacting
with a technology system, but a set of humans, interacting
with one another as well as that underlying technology system~\citep{event13laur}.
This extension can be viewed as an ``event-driven" non-cooperative
game theoretic approach, which is distinct from
both differential games (in which player moves are real-valued,
and chosen continually, at all times) and Markov games (which
lack hidden variables). Future work involves
investigating this event-driven game theory and its application to
likelihood ratio based attack detection.  

\section{Acknowledgments}
The authors would like to acknowledge the Army Research Office (ARO) for supporting this work under grant number W911NF15-1-0127.  

\bibliographystyle{apalike}
\bibliography{refs, cyber_refs}

\end{document}